\documentclass[aps,physrev,reprint,groupedaddress,nobibnotes,longbibliography]{revtex4-2}

\usepackage{graphicx}
\usepackage{amssymb, amsmath, latexsym}
\usepackage{amsthm}
\usepackage{physics}
\usepackage{nicefrac}
\usepackage{dsfont}
\usepackage{pifont}
\usepackage[]{mdframed}
\usepackage{thm-restate}
\usepackage{datetime}
\newdate{date}{28}{5}{2026}
\usepackage{tikz}
\usetikzlibrary{backgrounds,fit,decorations.pathreplacing,calc}

\usepackage{bbm}

\newtheorem{theorem}{Theorem}
\newtheorem{lemma}{Lemma}

\newtheorem{definition}{Definition}
\newtheorem{corollary}{Corollary}[theorem]
\newtheorem{remark}{Remark}

\newcommand*{\sbin}{\{0,1\}}
\newcommand{\kron}{\otimes}
\newcommand{\ptrs}[1]{\operatorname{tr}_{#1}}
\newcommand{\id}{{\mathbbm{1}}}

\newcommand{\vecstate}[1]{\ket{#1}\bra{#1}}
\newcommand{\mE}{\mathcal E}
\newcommand{\idx}[2]{{#1}_{#2}}
\newcommand{\psiAB}{\ensuremath{\idx{\psi}{AB}}}
\newcommand{\mX}{\mathcal X}
\newcommand{\hi}[1]
{\ensuremath{\mathcal{H}_{\textnormal{#1}}}}
\newcommand{\proj}[1]{|#1\rangle\langle#1|}
\newcommand{\hA}{\hi{A}}
\newcommand{\hAB}{\hi{AB}}
\newcommand{\hB}{\hi{B}}
\newcommand{\rhoA}{\ensuremath{\idx{\rho}{A}}}
\newcommand{\rhoB}{\ensuremath{\idx{\rho}{B}}}
\newcommand{\rhoAB}{\ensuremath{\idx{\rho}{AB}}}
\newcommand{\idA}{\idi{A}}

\newcommand{\ptrace}[2]{\ensuremath{\ptr{#1} (#2)}}
\newcommand{\ptr}[1]{\operatorname{tr}_{\textnormal{#1}}}
\newcommand{\hR}{\hi{R}}
\newcommand{\hX}{\hi{X}}
\newcommand{\states}[1]{\ensuremath{\mathcal{S}(#1)}}
\newcommand{\posops}[1]{\ensuremath{\mathcal{P}(#1)}}
\newcommand{\h}{\ensuremath{\mathcal{H}}}

\newcommand{\idi}[1]{\ensuremath{\mathds{1}_{\textnormal{\tiny #1}}}}

\newcommand{\eps}{\varepsilon}
\newcommand{\dis}{D}

\begin{document}

\title{Impossibility of Quantum Private Queries}

\author{Esther H\"anggi}
\thanks{These authors contributed equally to this work.}
\affiliation{Lucerne School of Computer Science and Information Technology, Lucerne University of Applied Sciences and Arts, Rotkreuz, Switzerland}

\author{Severin Winkler}
\thanks{These authors contributed equally to this work.}
\affiliation{Ergon Informatik AG, Zurich, Switzerland}

\date{\displaydate{date}}

\begin{abstract}
Symmetric private information retrieval is a cryptographic task allowing a user to query a database and obtain exactly one entry without revealing to the owner of the database which element was accessed. The task is a variant of general two-party protocols called \textit{one-sided secure function evaluation} and is closely related to oblivious transfer.
Under the name \emph{quantum private queries}, quantum protocols have been proposed to solve this problem in a cheat-sensitive way: In such protocols, it is not impossible for dishonest participants to cheat, but they risk detection [V.\ Giovannetti, S.\ Lloyd, and L.\ Maccone, Phys.\ Rev.\ Lett.\ 100, 230502 (2008)].
We give an explicit attack against any cheat-sensitive symmetric private information retrieval protocol, showing that any  protocol that is secure for the user cannot have non-trivial security guarantees for the owner of the database: The user is even able to retrieve the \emph{complete} database.
\end{abstract}

\maketitle

\emph{Introduction.} Quantum physics enables the implementation of cryptographic tasks which are impossible to realize with access to classical communication only. Key distribution~\cite{bb84,ekert} and random-number generation~\cite{qrng,Jennewein2000} are certainly the most prominent examples that are already available commercially. 

From the early days of quantum cryptography, protocols for tasks between two mistrustful parties have also been developed. Wiesner~\cite{wiesner} described ``A means for transmitting two messages either but not both of which may be received,'' which corresponds to oblivious transfer. In their seminal work~\cite{bb84}, Bennett and Brassard proposed not only a quantum key distribution protocol, but also a protocol for coin tossing and implicitly a protocol for bit commitment. Subsequently, quantum protocols for bit commitment and oblivious transfer have been proposed in Refs.~\cite{BC91,BBCS92,BCJL93,Ard95}.

In this work, we study quantum two-party protocols that allow a user to obtain a single entry of their choice from a remote database without revealing to the owner of the database which item they are interested in. This task is known as (single-database) \textit{symmetric private information retrieval}~\cite{CGKS95} and is closely related to a variant of oblivious transfer called 1-out-of-$n$ oblivious transfer.

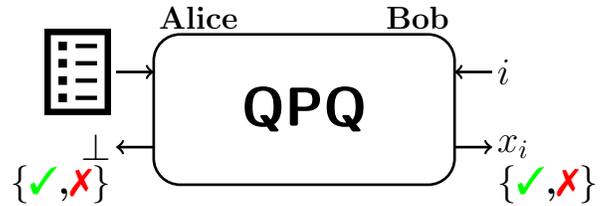
\begin{figure}[h]
\centering
\begin{tikzpicture}	[scale=0.95,every node/.style={inner sep=1,outer sep=1}]
\node[anchor=east] at (-2.5,1.5){\phantom{$\sum\limits_{x \in \mX_i}\frac{1}{\sqrt{|\mX|}}\ket{xx}$}};
\node[anchor=south west] at (-2,2){\textbf{\large Alice}};
\node[anchor=south east] at (2,2){\textbf{\large Bob}};
\draw[line width=1pt,->](-2.5,1.5)--(-2,1.5);
\node[anchor=east] at (-2.5,1.5){\includegraphics[width=25pt]{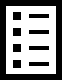}};
\draw[line width=1pt,<-](-2.5,0.5)--(-2,0.5);
\node[anchor=east] at (-2.5,0.5){\textbf{\Large $\bot$}};
\node[anchor=east] at (-2.5,0){\textbf{\Large $\{\text{\textcolor{green}{\ding{51}},\textcolor{red}{\ding{55}}} \}$}};
\draw[line width=1pt,->](2.5,1.5)--(2,1.5);
\node[anchor=west] at (2.5,1.5){\textbf{\Large $i$}};
\draw[line width=1pt,<-](2.5,0.5)--(2,0.5);
\node[anchor=west] at (2.5,0.5){\textbf{\Large $x_i$}};
\node[anchor=west] at (2.5,0){\textbf{\Large $\{\text{\textcolor{green}{\ding{51}},\textcolor{red}{\ding{55}}} \}$}};
\draw[rounded corners=10,line width=1pt](0,2)--(-2,2)--(-2,0)--(2,0)--(2,2)--(0,2);
\node at (0,1){\huge{\sf{\textbf{QPQ}}}};
\end{tikzpicture}
    \caption{Quantum private queries. Alice, the owner of the database, holds a database of $n$ entries $x_1,\ldots,x_n \in \mX^n$. Bob, the user of the database, can send a query $i$ to the database and receives the corresponding database entry $x_i$. If either party tries to cheat and learn more than what is allowed, the other party has some chance to detect this. This is modeled by an additional output from $\{\texttt{accept},\texttt{reject}\}$ for both players.}\vspace{-0.2cm}
    \label{fig:cheat-qpq}
\end{figure}

It was later discovered that many quantum two-party primitives cannot be implemented securely without further assumptions~\cite{Mayers97,LoChau97,Lo97,Kitaev,Ambain01,bcs12,Colbec07}. 
However, quantum protocols still help to limit the degree by which a dishonest participant can violate the security requirements. They can achieve strictly better security bounds than any classical protocol~\cite{SpeRud01,chailloux2009optimal,CGS16,ATVY00,Mochon04,Colbec07}.

\emph{Cheat-sensitive quantum protocols.} A cheat-sensitive setting considers dishonest players who want to avoid being caught cheating. This is a natural model in commercial or social situations where cheaters can be punished~\cite{Franklin92,Aumann2010}. Protocols then do not have to prevent cheating, but must detect violations of the security guarantees with a certain probability. That quantum physics may allow for cheat detection can be seen in the example of quantum key distribution, where eavesdropping is not impossible, but an eavesdropper will be revealed. Even for cryptographic primitives that cannot be implemented securely by quantum protocols, cheat-sensitive versions may still be possible.

In two-party computations, many quantum weak coin tossing protocols can be interpreted as implementing cheat detection, with the party that was detected cheating losing the coin toss~\cite{SpeRud02}. Cheat-sensitive weak coin tossing protocols have even been implemented experimentally~\cite{Neves2023}.

Quantum bit escrow~\cite{ATVY00,hardykent04} is a primitive related to bit commitment with one-sided cheat-sensitivity, which can be achieved using quantum protocols.  
Cheat-sensitive string commitment can be implemented with a security condition that guarantees the receiver to be caught if they learn too much information in the holding phase~\cite{bchlw08}.

\emph{Quantum private qeries.} In Ref.~\cite{GLM08}, the authors propose a cheat-sensitive version of symmetric private information retrieval under the name \emph{quantum private queries} (Figure~\ref{fig:cheat-qpq}). The protocol has been analyzed in Ref.~\cite{GLM10} and implemented in Ref.~\cite{demartini09}.  Similar functionalities have also been proposed in Refs.~\cite{JSGBBWZ11,Olejnik11}. The protocol~\cite{JSGBBWZ11} has been implemented experimentally in Ref.~\cite{Chan2014}.

This work investigates whether there exists a secure protocol for quantum private queries. We answer this question in the negative, demonstrating that securely implementing quantum private queries is impossible, even when tolerating some (small) error $\eps$ and when weakening the security guarantees for the database owner. We first present a specific attack on the protocol in~\cite{GLM08,GLM10}. Then we give a generic attack on \emph{any} quantum private queries protocol. 

In Ref.~\cite{JSGBBWZ11}, it is argued that the impossibility results of Ref.~\cite{Lo97} do not apply because the user privacy of their protocol is not perfect, but only holds with high probability. Indeed, tolerating a small error in the implementation can in general help to build secure two-party protocols~\cite{WW10} that are provably impossible in the perfect case~\cite{SaScSo09}. However, the impossibility result that follows from our generic attack also rules out secure quantum private query protocols in the imperfect case, where correctness and security only hold approximately.

\emph{Attack on protocol~\cite{GLM08}.} The protocol proposed in Ref.~\cite{GLM08}  works as follows:  
The user of the database (Bob) sends either of two query states $\ket{i}$ or $(\ket{0}+\ket{i})/\sqrt{2}$ to the owner of the database (Alice). Which one of the two states is sent first is chosen at random. Alice does a quantum database query and sends back the state $\ket{i}\ket{x_i}$, where $x_i$ is the entry of the database. Bob can retrieve the database entry by measuring the state $\ket{i}\ket{x_i}$ and then check whether Alice has been cheating by measuring whether the other state is in the correct superposition. 

In the following, we consider an answer to a query \emph{valid} if Bob accepts it. In Ref.~\cite{GLM10}, the authors analyzed the protocol under the assumption that every query has a \emph{unique} valid answer. This condition can be met by requiring that each query admits a unique correct answer that Bob can independently verify. Under this assumption, it is not possible to reach information-theoretic security since a computationally unbounded Bob could reconstruct the database exactly as Alice does, simply by computing all valid answers. The authors of Ref.~\cite{GLM08} suggest that the database could store the solutions to a computational problem, which are hard to compute but efficiently verifiable, such as the preimages of a computationally secure one-way function. Bob could then compute the answer to the problem himself, but not efficiently. 
This functionality could \emph{a priori} be of interest. However, oblivious transfer, and thus symmetric private information retrieval, can be constructed from a quantum one-way function~\cite{BCKM21}. The result in Ref.~\cite{BCKM21} even holds if only the user of the database is computationally bounded. Under the assumption that quantum one-way functions exist, it is, therefore, not necessary to introduce cheat-sensitivity and resort to a protocol like quantum private queries.

Thus, in our view, the only compelling setting for quantum private queries arises when database queries may have multiple valid answers. Once the assumption of \emph{unique} valid answers is dropped, Giovannetti et al.~\cite{GLM10} showed that their protocol becomes vulnerable to an attack. To mitigate this specific attack, they proposed a variant in which the user preserves security by repeatedly submitting the same query. However, they left open the question of whether a secure quantum private queries protocol can exist under these additional assumptions or in general.

Our first attack below targets the specific protocol from Ref.~\cite{GLM08}. It is stronger than the one presented in Ref.~\cite{GLM10}: Even a database user who is able to verify all received answers cannot detect this attack. It also allows the database owner to keep answering consistently in a protocol in which the user repeatedly asks the same query to maintain security. Therefore, the countermeasures proposed in Ref.~\cite{GLM10} cannot prevent this attack.

\begin{figure}[h!]
\centering
\begin{tikzpicture}	[scale=0.95,every node/.style={inner sep=1,outer sep=1}]
\node[anchor=south west] at (-2,2){\textbf{\large Alice}};
\node[anchor=south east] at (2,2){\textbf{\large Bob}};
\draw[line width=1pt,->](-2.5,1.5)--(-2,1.5);
\node[anchor=east] at (-2.5,1.5){$\sum\limits_{x \in \mX_i}\frac{1}{\sqrt{|\mX|}}\ket{xx}$};
\draw[line width=1pt,<-](-2.5,0.5)--(-2,0.5);
\node[anchor=east] at (-2.5,0.5){$ \dis(\rho_A^i,\rho_A^{j})$};
\draw[line width=1pt,->](2.5,1.5)--(2,1.5);
\node[anchor=west] at (2.5,1.5){\textbf{\Large $i$}};
\draw[line width=1pt,<-](2.5,0.5)--(2,0.5);
\node[anchor=west] at (2.5,0.5){\textbf{\Large $x_i,\text{\textcolor{green}{\ding{51}}}$}};
\draw[rounded corners=10,line width=1pt](0,2)--(-2,2)--(-2,0)--(2,0)--(2,2)--(0,2);
\node[anchor=east](A1) at (2,1.7){\color{gray}{$\ket{i}$}};
\node(A2) at (-1.5,1.7){};
\node[anchor=west](B1) at (-2,1.3){\color{gray}{$\ket{i}\ket{x_i}$}};
\node(B2) at (1.5,1.3){};
\draw[line width=1pt,->,color=gray](A1) edge (A2);
\draw[line width=1pt,->,color=gray](B1) edge (B2);
\draw[rounded corners=10,line width=1pt](0,2)--(-2,2)--(-2,0)--(2,0)--(2,2)--(0,2);
\node[anchor=east](C1) at (2,0.7){\color{gray}{$\frac{\ket{i}{+}\ket{0}}{\sqrt{2}}$}};
\node(C2) at (-1.5,0.7){};
\node[anchor=west] (D1) at (-2,0.3){\color{gray}{$\frac{\ket{i}\ket{x_i} {+}\ket{0}\ket{d}}{\sqrt{2}}$}};
\node(D2) at (1.5,0.3){};
\draw[line width=1pt,->,color=gray](C1) edge (C2);
\draw[line width=1pt,->,color=gray](D1) edge (D2);
\end{tikzpicture}
    \setlength{\belowcaptionskip}{-6pt}

    \caption{The protocol~\cite{GLM08} is depicted inside the box. Using the state $\ket{\phi^i}_{X_iX'_i} =\sum_{x \in \mX_i}\frac{1}{\sqrt{|\mX_i|}}\ket{xx}_{X_iX'_i}$ with the database values in the second register, a dishonest Alice can distinguish between different queries of Bob. Bob always accepts.}
    \label{fig:concrete_attack}
\end{figure}
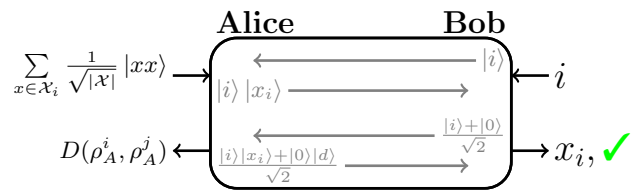

Note that in the case where the queries have multiple valid answers, the owner of the database can, in particular, select among the valid answers at random before the start of the protocol and then execute the protocol with the randomly chosen database entries. The random selection can be represented by a mixed state over valid answers. This can be achieved by creating $\ket{\psi}=\sum_x \sqrt{P_X(x)}\ket{x}\ket{x}$, treating one part of the state as holding the random value and keeping the other part.
\begin{mdframed}\textbf{Specific attack.} (See Figure~\ref{fig:concrete_attack}) Alice, the owner of the database, prepares states $\ket{\phi^i}_{X_iX'_i} {=}
\sum\limits_{x \in \mX_i}\frac{1}{\sqrt{|\mX_i|}}\ket{xx}_{X_iX'_i}$ for all $i\in\{1,\ldots,n\}$\ \footnote{Note that the database for the protocol~\cite{GLM08} uses an additional dummy entry at index $0$. Thus, here the $n$ actual database entries are $x_1,\ldots,x_n$.}, where $\mX_i$ is the set of valid answers for query $i$. She then executes the complete protocol in Refs.~\cite{GLM08,GLM10} with the first register as input while keeping the second register \footnote{This is an attack type often used in quantum two-party protocols and corresponds to the player keeping the purification of the state.}. For the user of the database, this strategy is indistinguishable from a uniform random selection. Since Alice chooses a superposition of valid strategies, this purified strategy is also allowed.
After the end of the protocol, she performs the Helstrom measurement~\cite{helstrom1976quantum} that optimally distinguishes between the marginal states corresponding any two queries $i \neq j \in \{1,\ldots,n\}$.
\end{mdframed}

Since the measurement happens after the protocol has ended, the attack cannot be detected by the other party. 
We consider the two marginal states of the database owner $\rho_A^i$ and $\rho_A^{j}$ that result from the execution of the protocol with any two queries $i \neq j \in \{1,\ldots,n\}$. Using the \emph{trace distance}, $D(\rho,\sigma):= \frac{1}{2}\| \rho-\sigma \|_1$, we obtain the following lower bound on the distinguishability of these two states:
\begin{equation*}
    \dis(\rho_A^i,\rho_A^{j}) \geq 1-\frac{1}{\max\{|\mX_i|,|\mX_j|\}}\;.
\end{equation*}
This violates the security definition in Ref.~\cite{GLM10}, which requires that no two queries can be distinguished. The bound is independent of the number of entries in the database. The mathematical details of the computation of the trace distance are given in the appendix.

\emph{Impossibility of quantum private queries.} We now show that there is an attack on any quantum private queries protocol. Our analysis adopts the standard model of quantum two-party computation~\cite{Mayers97,LoChau97,Lo97}, where the two players have access to a noiseless quantum and a noiseless classical channel. In each round of the protocol, one party may perform, conditioned on the available classical information, an arbitrary quantum operation on the system in their possession. This operation also generates the input for the available communication channels. The total system — comprising the subsystems controlled by Alice and Bob — is initially in a pure state. By introducing ancillary spaces, the quantum operations of both parties can be purified. Thus, we can assume that the state at the end of the protocol is pure, conditioned on all the classical communication.

\emph{Security conditions.} Any secure symmetric private information retrieval protocol must fulfill conditions for correctness, user privacy, and data privacy, at least with a statistical error $\eps \geq 0$. Cheat-sensitive protocols such as quantum private queries do not have to prevent dishonest players from violating the security conditions, but have to detect any violation with probability greater than $\delta \geq 0$. This leads to the following three conditions. Note that these conditions are necessary but not sufficient to show the security of a protocol. In particular, all three conditions necessarily hold for any protocol that is secure according to simulation-based definitions of cheat-sensitive security~\cite{Aumann2010} and secure two-party computation~\cite{FS09} (see also Ref.~\cite{WW10}). We will show that no protocol can fulfill these three necessary conditions; therefore, no secure protocol exists.

Correctness states that the protocol works as expected when both players are honest.
\begin{definition}[Correctness]\label{def:correctness}
Let Alice hold a database with entries $x_1,\ldots ,x_n \in \mX^n$, and let Bob have query $i \in \{1,\ldots,n\}$. 
A quantum private queries protocol is \emph{$\eps$-correct} if in any execution of the protocol where both players are honest:
\begin{equation}\label{eq:correctness}
    \Pr[\text{both players accept} \wedge Y = x_i] \geq 1-\eps
\end{equation}
where $Y$ is the output of Bob.
\end{definition}
Correctness of a protocol according to this definition implies that condition~\eqref{eq:correctness} still holds when the database owner chooses their entries randomly.

User privacy, or security for Bob, requires that the owner of the database cannot learn which query was made from their state $\rho_A^i$ at the end of an execution of the protocol with query $i$, i.e., the marginal state at the end of the protocol is the same for all queries. For $\eps$-security, all $\rho_A^i$ must therefore be $\eps$-close to a fixed state, which implies by the triangle inequality that the states $\rho_A^i$ and $\rho_A^j$ are at least $2\eps$-close.
The existence of such a fixed (ideal) state is, in particular, guaranteed by the existence of a simulator in the real-ideal-setting (see Refs.~\cite{FS09,WW10,HänggiWinkler2024}).
\begin{definition}[Weak User Privacy]\label{def:weak-user-privacy}
A quantum private queries protocol is \emph{weakly $(\eps,\delta)$-secure for the user} of the database if for any strategy of a dishonest Alice that Bob accepts with probability at least $1-\delta$, and any queries $i,j \in \{1,\ldots,n\}$,
\begin{equation}\label{eq:weak-user-privacy}
\dis(\rho_{A}^i, \rho_A^j) \leq 2\eps\;.
\end{equation}
\end{definition}

Finally, a secure quantum private queries protocol prevents the user of the database from obtaining the answers to more than one query from the marginal state at the end of the protocol $\rho_B$.  The fact that the entries of the database are \emph{a priori} unknown to the user is captured by choosing the database entries $X_1$,\ldots, $X_{n}$ uniformly at random from the set of valid answers $\mX_i$ for all $i \in \{1,\ldots,n\}$. For the following definition, we consider attacks
where the user tries to guess two entries $X_i$ and
$X_j$ for fixed indices $i \neq j \in \{1,\ldots,n\}$ by any measurement on his part of the state at the end of the
protocol. In this case, the security condition requires that
the user cannot do significantly better than guessing
one of the two entries.
\begin{definition}[Weak Data Privacy]\label{def:weak-database-privacy}
A quantum private queries protocol is \emph{weakly $(\eps,\delta)$-secure for the database owner} if for any strategy of a dishonest Bob that Alice accepts with probability at least $1-\delta$
\begin{align*}
    \max_{\mE_B}\Pr[\mE_B(\rho_{X_1,\ldots, X_n B}) = (X_i, X_j)] \leq \frac{1}{k}+\eps\;,
\end{align*}
where the maximization is over all measurements on system $B$ of the state $\rho_{X_1,\ldots, X_nB}$ at the end of the protocol that output a pair of values $(Y, Y') \in \mX$, $k:=\min\{|\mX_i|,|\mX_j|\}$ and $i \neq j \in \{1,\ldots,n\}$.
\end{definition}

Correctness according to Definition~\ref{def:correctness} implies that any execution of the protocol where the dishonest player follows the purified protocol is accepted by the honest player with probability at least $1-\eps$. Thus, for any such execution of a protocol, the conditions for user and data privacy must hold with $\delta = \eps$. We use this fact in the proof of our main result to show that there is no quantum private queries protocol for any $\eps \leq \delta \leq 1$ \footnote{A stronger correctness condition would require that both players always accept the protocol if it is honestly executed. Since this stronger requirement satisfies our correctness condition, our impossibility result also holds in this case, for any $0\leq \delta \leq 1$.}. Note that there is a trivial quantum private queries protocol that is secure if $\delta < \eps$ because an honest execution of the protocol can be rejected with probability larger than $\delta$ without violating correctness \footnote{In this trivial protocol, Alice sends all $n$ database entries $x_i$ to Bob, who retrieves the database entry of his choice and always accepts. Alice rejects the protocol with probability $\eps$ independently of everything else. This protocol is correct and obviously secure for Bob. Since no strategy of Bob is accepted with probability at least $1-\delta >1-\eps$, Definition~\ref{def:weak-user-privacy} becomes trivial for $0 \leq \delta < \eps$.}. We will use the notation \emph{$\eps$-secure} for $(\eps,\eps)$-secure.

A primitive related to quantum private queries is \textit{private information retrieval}~\cite{classicalpir}. It does not guarantee data privacy, so a dishonest user may retrieve an arbitrary number of entries. The task admits a trivial and secure solution: The server simply sends the entire database to the user. This naturally raises the question of whether more communication-efficient protocols are possible. The trivial protocol is known to be optimal in terms of communication complexity under perfect security~\cite{classicalpir}, even when quantum communication is allowed~\cite{qpir_nayak}. Quantum protocols with sublinear communication complexity have been proposed~\cite{legall}, but these protocols were subsequently shown to be insecure: Any quantum protocol for private information retrieval must exhibit linear communication complexity, even in the case of approximate user privacy and approximate correctness~\cite{Baumeler2015}.

\emph{A generic attack on quantum private queries.} We now construct an attack on any implementation of quantum private queries. This negatively answers the questions posed in Refs.~\cite{GLM10,Onofri2024} whether \emph{any} secure quantum private queries protocol exists. 

We show that for any quantum private queries protocol that is $\eps$-correct and $\eps$-secure for the user, there exists an attack that allows the user to compute $m$ database entries for any $2 \leq m \leq n$ with high probability, contradicting weak data privacy. 

The attack follows the intuition of Lo's initial impossibility proof~\cite{Lo97} that there exists a local operation for the user of the database to recover any other entry. Our analysis provides a quantitative bound that also applies in the non-asymptotic case, where the error is a constant. 

\begin{figure}[ht!]
\centering
\begin{tikzpicture}	[scale=0.95,every node/.style={inner sep=1,outer sep=1}]
\node[anchor=east] at (-2.5,1.5){\phantom{$\sum\limits_{x \in \mX_i}\frac{1}{\sqrt{|\mX|}}\ket{xx}$}};
\node[anchor=west] at (2.5,0){\phantom{\textbf{\Large $\{\text{\textcolor{green}{\ding{51}},\textcolor{white}{\ding{55}}} \}$}}};
\node[anchor=south west] at (-2,2){\textbf{\large Alice}};
\node[anchor=south east] at (2,2){\textbf{\large Bob}};
\draw[line width=1pt,->](-2.5,1.5)--(-2,1.5);
\node[anchor=east] at (-2.5,1.5){\includegraphics[width=25pt]{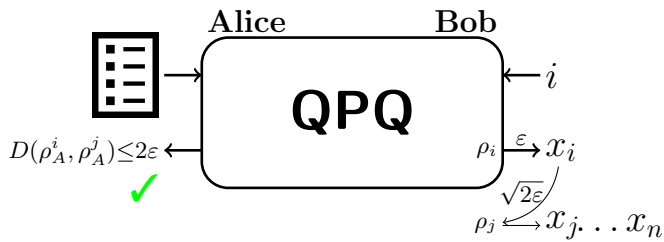}};
\draw[line width=1pt,<-](-2.5,0.5)--(-2,0.5);
\node[anchor=east] at (-2.5,0.5){$\dis(\rho_A^i,\rho_A^{j}){\leq}2\eps$};
\node[anchor=east] at (-2.5,0){\textbf{\Large $\text{\textcolor{green}{\ding{51}}} $}};
\draw[line width=1pt,->](2.5,1.5)--(2,1.5);
\node[anchor=west] at (2.5,1.5){\textbf{\Large $i$}};
\draw[line width=1pt,<-](2.5,0.5)--(2,0.5);
\node[anchor=south] at (2.25,0.5){$\eps$};
\node[anchor=west](xi) at (2.5,0.5){\textbf{\Large $x_i$}};
\draw[rounded corners=10,line width=1pt](0,2)--(-2,2)--(-2,0)--(2,0)--(2,2)--(0,2);
\node at (0,1){\huge{\sf{\textbf{QPQ}}}};
\node[anchor=east](rhoi) at (2.0,0.5){$\rho_i$};
\node[anchor=east](rhoj) at (2.0,-0.5){$\rho_j$};
\node[anchor=west](xj) at (2.5,-0.5){\textbf{\Large $x_j$}};
\node[anchor=west] at (2.9,-0.5){\textbf{\Large $\ldots x_n$}};
\draw[->](xi) to [bend left=35]  (rhoj);
\node[anchor=south] at (2.25,-0.3){$\sqrt{2\eps}$};
\draw[->](rhoj) to  (xj);
\end{tikzpicture}
    \caption{Our generic attack on any protocol implementing quantum private queries. Bob follows the protocol honestly except from keeping his state purified and Alice can therefore not detect the attack. If the protocol is $\eps$-correct and $\eps$-secure for the user, then the user can step by step retrieve all database entries with high probability.}
    \label{fig:general_attack}
\end{figure}

\begin{mdframed}[nobreak=true]\textbf{Generic attack.}(See Figure~\ref{fig:general_attack}) Bob, the user of the database, executes the protocol honestly with query $i$.
After the end of the protocol, he computes step by step all database entries: In each step, he applies a measurement to compute the next database entry and then rotates his state by a local unitary transformation.
\end{mdframed}

\begin{figure}[ht]
\centering
\begin{tikzpicture}	[line width=1pt]
    \node(rho1)[align=center, anchor=center] at (0,0) {$\rho_1$};
    \node(rho2)[align=center, anchor=center] at (2,0) {$\rho_2$};   
    \node(rho3)[align=center, anchor=center] at (4,0) {$\rho_3$};
    \node(rhol)[align=center, anchor=center] at (7,0) {$\rho_l$};
    \node(rho2tilde)[align=center, anchor=center] at (2,1.5) {$\bar{\rho}_2$};   
    \node(rho3tilde)[align=center, anchor=center] at (4,2) {$\bar{\rho}_3$};
    \node(rholtilde)[align=center, anchor=center] at (7,2.5) {$\bar{\rho}_l$};
    \node(X1)[align=center, anchor=center] at (0,-2) {$X_1$};
     \node(Prho1)[align=center, anchor=west] at (0-0.85,-2.5) {$\Pr[X_1]{\geq} 1{-}\varepsilon$};
    \node(X2)[align=center, anchor=center] at (2,-2) {$X_2$};
    \node(Prho2)[align=center, anchor=west] at (2-0.85,-2.5) {$\Pr[X_2]{\geq} 1{-}\varepsilon$};  
    \node(X3)[align=center, anchor=center] at (4,-2) {$X_3$};
    \node(Prho3)[align=center, anchor=west] at (4-0.85,-2.5) {$\Pr[X_3]{\geq} 1{-}\varepsilon$};
    \node(Xl)[align=center, anchor=center] at (7,-2) {$X_l$};
    \node(Prhol)[align=center, anchor=west] at (6.5-0.85,-2.5) {$\Pr[X_l]{\geq} 1{-}\varepsilon$};
    \draw[->, line width=0.2mm, color=black] (rho1) -- (X1) node[pos=0.3, anchor = west] {$\mathcal{E}_1$};
    \draw[->, line width=0.2mm, color=black] (rho2) -- (X2) node[pos=0.3, anchor = west] {$\mathcal{E}_2$};
    \draw[->, line width=0.2mm, color=black] (rho3) -- (X3) node[pos=0.3, anchor = west] {$\mathcal{E}_3$};
    \draw[->, line width=0.2mm, color=black] (rhol) -- (Xl) node[pos=0.3, anchor = west] {$\mathcal{E}_l$};
        \draw[|-|, line width=0.2mm, color=red] (rho2) -- (rho2tilde) node[pos=0.5, anchor = south, rotate=-90, yshift=0.2cm, fill=white!30] {$D\leq 3\sqrt{\varepsilon}{+}\varepsilon$};
    \draw[|-|, line width=0.2mm, color=red] (rho3) -- (rho3tilde) node[pos=0.5, anchor = south,rotate=-90, yshift=0.2cm,fill=white!30] {$D\leq 2{\cdot} (3\sqrt{\varepsilon}{+}\varepsilon)$};
    \draw[|-|, line width=0.2mm, color=red] (rhol) -- (rholtilde) node[pos=0.5, anchor = south,rotate=-90, yshift=0.2cm,fill=white!30] {$D\leq (l-1){\cdot} (3\sqrt{\varepsilon}{+}\varepsilon)$};
    \begin{scope}[on background layer]
    \path[->,  line width=0.2mm, bend right=13]
     (X1) edge node[pos=0.3, anchor = west] {$U_{1,2}$} (rho2tilde);
      \path[->,  line width=0.2mm,bend right=13]
     (X2) edge node[pos=0.3, anchor = west] {$U_{2,3}$} (rho3tilde);
      \path[dashed,->,  line width=0.2mm,bend right=13]
     (X3) edge node[pos=0.3, anchor = west] {} (rholtilde);
    \end{scope}
    \end{tikzpicture}
       \caption{\label{fig:seq-operations}The generic attack on any quantum private queries protocol. The sequence of rotations and measurements are all applied after the end of the protocol. We define the state after the $i$th step of the attack $\bar{\rho}_{i+1}:=S_{i+1}(S_{i}(\ldots S_2(\rho_{AB}^1)))$, where each operation $S_i$ corresponds to the measurement succeeded by a unitary transformation. The gentle measurement lemma gives an upper bound on the disturbance caused by the measurement. By Uhlmann's theorem, there is a unitary $U_{i-1,i}$ acting on Bob's system only that rotates the state $\rho_{i-1}$ to $\rho_i$. We show that adding disturbance of the measurement and the imprecision of the rotation introduces an error of at most $3\sqrt{\eps}+\eps$ in terms of the trace distance in each step.}
\end{figure}
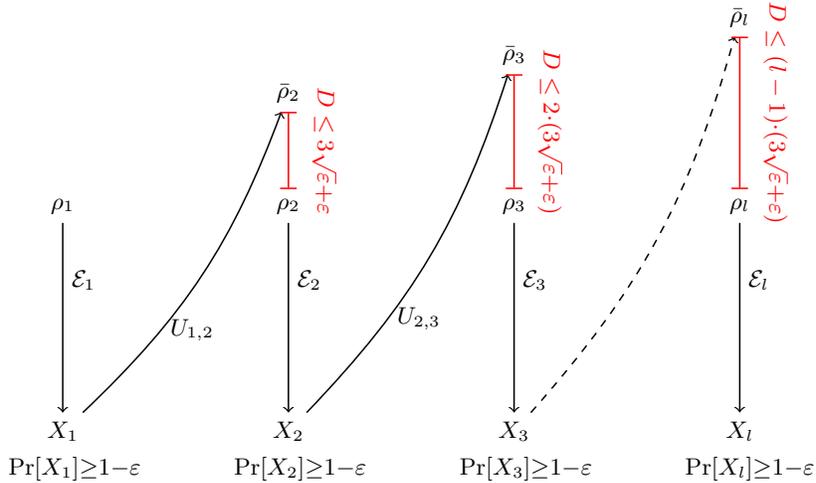

We can put a quantitative lower bound on the success probability of this attack as follows (see Figure~\ref{fig:seq-operations}): The $\eps$-correctness of the protocol implies that Bob has a measurement that outputs the correct database entry with high probability, at least $1-\eps$. By the gentle measurement lemma~\cite{Winter99,Wilde13}, this measurement does not disturb the state much, and the post-measurement state is $\sqrt{\eps}$-close. Weak user privacy guarantees that the two marginal states of Alice resulting from an execution of the protocol with query $i$ and $j$ are $2\eps$-close. By Uhlmann's theorem~\cite{Uhlman76}, this implies that there exists a unitary that rotates Bob's state $\rho_B^i$ to a state $2\sqrt{\eps}$-close to the state $\rho_B^j$ corresponding to another index $j$. From this state, the database entry $x_{j}$ can be obtained (approximately) by measuring. 

The measurement followed by the rotation defines a quantum operation $S_{i,j}$ that computes the $i$th database entry and rotates Bob's state close to $\rho_B^j$. This operation introduces an error of at most $3\sqrt{\eps}+\eps$ in terms of the trace distance. The total disturbance caused to Bob's state by successively applying $1 \leq k \leq n-1$ of these operations can be bounded by $k\times (3\sqrt{\eps}+\eps)$, i.e., by the sum of the disturbances caused by the individual operations. Using this fact, we can bound the probability for each step that Bob's computation fails to determine the correct database entry. By adding these individual errors, we obtain the probability to correctly compute \emph{all} entries as stated in the theorem below. 

Since the rotations and measurements happen after the protocol has ended, the attack cannot be detected by the other party.

The complete proof of the following theorem with all the calculations is given in the appendix. 
\begin{restatable}{theorem}{mainresult}\label{thm:main-result}
For any $\eps \leq \delta \leq 1$ and any quantum private queries protocol that is $\eps$-correct and $(\eps,\delta)$-secure for the user of the database, a dishonest user can retrieve $2\leq m \leq n$ database entries with probability at least
\begin{equation*}
       1-2m^2\sqrt{\eps}\;.
\end{equation*}
\end{restatable}

This result also immediately implies the impossibility of protocols without cheat-sensitivity, since security must hold against any strategy according to a non-cheat-sensitive security definition, i.e., in the so-called malicious model. 

The proof of Theorem~\ref{thm:main-result} implies that there is no quantum private queries protocol for databases with $n \geq 2$ entries that is $\eps$-correct and weakly $\eps$-secure for both the owner and the user of the database for any $\eps \leq \frac{1}{64}$ (see Corollary~1.1 in the appendix). Since Bob can potentially learn \emph{all} entries with significant probability, the attack of Theorem~\ref{thm:main-result} can violate any sensible definition of security for Alice, even a very weak definition that only prevents a dishonest Bob from learning all database entries. In particular, this excludes the existence of any protocol that is secure according to the security condition considered in Ref.~\cite{GLM08}, which allows a dishonest user to obtain the answers to two queries.

Theorem~\ref{thm:main-result} excludes the possibility of secure quantum private query protocols for any database admitting at least two queries with multiple valid answers. In all remaining cases, security can only be guaranteed against computationally bounded users, and this can be achieved without relying on cheat-sensitive protocols (see Corollary~1.1 and Remarks~1 and 2 in the appendix).

The impossibility result of Theorem~\ref{thm:main-result} also holds for the secure evaluation of any function and, in particular, for different variants of oblivious transfer~\cite{HänggiWinkler2024}.

The proof techniques developed for this result can be used to prove results on the efficiency of protocols that implement a two-party function from oblivious transfer or other resource primitives~\cite{HänggiWinkler2024}. For example, it can be shown that securely implementing an instance of 1-out-of-$n$ oblivious transfer requires at least approximately $n-1$ instances of 1-out-of-2 oblivious transfer, also when using quantum protocols. 

\emph{Discussion.} We have shown the impossibility of secure quantum protocols for cheat-sensitive symmetric private information retrieval by giving an explicit attack on any such protocol. This generic attack can be applied even when the functionality is only approximated and even if the security guarantees for the owner of the database are strongly relaxed. It implies that any theoretical protocol proposed for quantum private queries or experimental implementation can only be secure against a restricted adversary. Finding the minimal physical or computational assumptions that enable quantum symmetric private information retrieval is an open problem to be answered in future research. It would also be interesting to investigate whether there exist assumptions allowing secure protocols only in the cheat-sensitive setting, but not in the non–cheat-sensitive model.

\begin{acknowledgments}
\emph{Acknowledgments.}|\,The authors would like to thank J\"urg Wullschleger for helpful discussions and \"Amin Baumeler, Lorenzo Laneve, and Stefan Wolf for comments on an earlier draft of this paper. This work was supported by the Swiss National Science Foundation Practice-to-Science Grant No. 199084.
\end{acknowledgments}

\bibliography{refs}

\begin{thebibliography}{56}
\bibitem [{a()}]{a}%
  \BibitemOpen
  \bibinfo {note} {Note that the database for the protocol~\cite {GLM08} uses
  an additional dummy entry at index $0$. Thus, here the $n$ actual database
  entries are $x_1,\protect \ldots ,x_n$.}\BibitemShut {Stop}%
\bibitem [{b()}]{b}%
  \BibitemOpen
  \bibinfo {note} {This is an attack type often used in quantum two-party
  protocols and corresponds to the player keeping the purification of the
  state.}\BibitemShut {Stop}%
\end{thebibliography}

\onecolumngrid
\newpage
\appendix

\setcounter{page}{1}

\section{Impossibility of Quantum Private Queries --- Appendix}

\subsection{Preliminaries}
We refer to a standard textbook on quantum information such as~\cite{NieChu00} for a detailed introduction. 

We consider finite-dimensional Hilbert spaces $\hi{}$. $\posops{\hi{}}$ denotes the set of positive semi-definite operators on~$\hi{}$; $\posops{\hi{}}:\hi{}\rightarrow \hi{}$ .
The set of quantum states $\states{\hi{}}$ is represented by positive semi-definite operators of trace one, $\states{\h} := \{ \rho \in \posops{\h} : \tr\,\rho = 1 \}$. 
We call a state $\rho \in \states{\h}$ \emph{pure} if it has rank one. A pure state $\rho$ can be represented as a ket $\ket{\psi}$, where $\ket{\psi}$ is an element of the Hilbert space $\h$, and the associated quantum state $\rho$ is the outer product $\rho = \proj{\psi}$. A state that is not pure is called \emph{mixed}.  
We will often consider bi- or more-partite quantum states. We use indices to denote multipartite Hilbert spaces, i.e., $\hAB:=\hA\kron\hB$. Given a quantum state $\rhoAB \in \states{\hA\kron\hB}$ we denote by $\rhoA$ and $\rhoB$ its marginal states $\rhoA=\ptrace{B}{\rhoAB}$ and $\rhoB=\ptrace{A}{\rhoAB}$. The symbol $\idA$ denotes either the identity operator on $\hi{A}$ or the identity operator on $\posops{\hA}$; it should be clear from the context which one is meant. 

We denote the distribution of a random variable $X$ by $P_X(x)$. A (classical) probability distribution $P_X$ of a random variable $X$ over $\mX$ can be represented by a quantum state $\rho_X$ using 
an orthonormal basis $\{\ket{x} \mid x\in\mX\}$ of a Hilbert space $\hi{$\mX$}$. The quantum state representing $P_X$ is  $\rho_X=\sum_{x\in\mX}P_X (x) \vecstate{x}$. 
Quantum information about a classical random variable $X$ can be represented by a classical-quantum state or cq-state $\rho_{XB}$ on $\hi{$\mX$}\kron\hi{B}$ of the form $\rho_{XB}=\sum_{x \in\mX}p_x\vecstate{x}\kron\rho_B^x$. 

A mixed state can be seen as part of a pure state on a larger Hilbert space. More precisely, given $\rho_A$ on $\hA$, there exists a pure density operator $\rho_{AB}$ on a
joint system $\hA \kron \hB$ such that 
\begin{equation}\label{eq:purification}
    \rho_A = \ptrace{B}{\rho_{AB}}\;.
\end{equation}
A pure density operator $\rho_{AB}$ for which~\eqref{eq:purification} holds is called a purification of $\rho_A$. For a state representing a classical distribution $\rho_X=\sum_{x\in\mX}P_X (x) \vecstate{x}$, one explicit purification is $\ket{\psi}_{XX'}=\sum_{x\in\mX} \sqrt{P_X (x)} \ket{x}_{X}\kron\ket{x}_{X'}$.

A transformation of a quantum system is represented by a trace-preserving completely positive map (TP-CPM). A \emph{measurement} $\mE$ can be seen as a TP-CPM that takes a quantum state on system $A$ and maps it to a classical register $X$, which contains the measurement result, and a system $A'$, which contains the post-measurement state. Such a measurement map is of the form $\mE(\rhoA)=\sum_x \proj{x}_X\otimes M_x \rho_{A} M_x^\dagger$, where all $\ket{x}$ are from an orthonormal basis of $\hX$ and the $M_x$ are such that $\sum_x M_x^{\dagger}M_x=\id_A$. The probability to obtain the measurement result $x$ is $\tr (  M_x \rho_{A} M_x^\dagger)$. Furthermore, the state of $A'$ after result $x$ is $M_x \rho_{A} M_x^\dagger/\tr (  M_x \rho_{A} M_x^\dagger)$. A \emph{projective} measurement is a special case of a quantum measurement where all measurement operators $M_x$ are orthogonal projectors. If the post-measurement state is not important, and we are only interested in the classical result, the measurement can be fully described by the set of operators $\{ \mE_x\}_x=\{ M_x^\dagger M_x\}_x$. Such a set is 
called a positive operator valued measure (POVM) and we call its elements $\mE_x$ \emph{elements of a measurement}. According to~\cite{Stine55}, any TP-CPM has a Stinespring dilation~\cite{Stine55}, i.e., there exists a Hilbert space $\hR$, a unitary $U$ acting on $\h_{AXA'R}$ and a pure state $\sigma_{XA'R}\in \states{\h_{XA'R}}$ with $\mE(\rhoA)=\ptrace{AR}{U(\rhoA \otimes \sigma_{XA'R}) U^\dagger}$.

\begin{definition}
The \emph{trace-distance} between two quantum states $\rho$ and $\sigma$ is 
\begin{align}
D(\rho,\sigma)&= \frac{1}{2}\| \rho-\sigma \|_1\;,
\end{align} 
where $\|A\|_1=\tr \sqrt{A^\dagger A}$. Two states $\rho$ and $\sigma$ with trace-distance at most $\delta$ are called \emph{$\delta$-close} and we write $\rho \approx_{\delta} \sigma$. 
\end{definition}
The trace-distance corresponds to the maximum distinguishing advantage in the following sense~\cite{NieChu00}: if one receives either $\rho$ or $\sigma$ at random, can make any measurement on the state and then has to guess which one it was, then the best strategy is correct with probability
\begin{equation}\label{eq:op-meaning-trace-distance}
   1/2+D(\rho,\sigma)/2\;.  
\end{equation}
For two cq-states $\rho^b_{XA}=\sum_{x \in \mX}P_X(x)\vecstate{x}\kron\rho^{b,x}_A$
with $b \in \sbin$, the trace distance fulfills the following property (see~\cite{Renner2005} for a proof)
\begin{equation}\label{eq:trace-distance-cq-states}
\dis(\rho^0_A,\rho^1_A)=\sum_{x \in \mX}P_X(x)\dis(\rho^{0,x}_A,\rho^{1.x}_A)\;.
\end{equation}

\subsection{Proof of the distinguishing advantage in the attack on the protocol~\cite{GLM08}}

In this part of the appendix, we give the mathematical details of the attack on the protocol in~\cite{GLM08}. This attack allows Alice to distinguish between different choices of Bob with high probability without Bob being able to detect this attack.

For completeness, we include a description of the original protocol for quantum private queries~\cite{GLM08}:
\begin{enumerate}
\item Alice creates a classical database $\{\ket{x_i}\}_i$ containing $n+1$ entries, where $x_i \in \mX$ for $i \in \{1,\ldots,n\}$ are the actual database entries and $\ket{x_0}=\ket{d}$ is an additional dummy entry. 
\item Bob selects the index $i\in \{1,\dots,n\}$ which he would like to query.
\item Bob tosses a coin $a$. If $a=0$, he prepares the states $\ket{\psi_1}=\ket{i}$ and $\ket{\psi_2}=(\ket{0}+\ket{i})/\sqrt{2}$. If 
$a=1$, then he prepares the states $\ket{\psi_1}=(\ket{0}+\ket{i})/\sqrt{2}$ and $\ket{\psi_2}=\ket{i}$.   
\item Bob sends $\ket{\psi_1}$ to Alice who performs a conditional database query, i.e., she applies the QRAM algorithm (see~\cite{NieChu00,GLM08} and the definition below) and returns the query and the register to Bob.
\item Bob sends $\ket{\psi_2}$ to Alice who performs a conditional database query and returns the query and the register to Bob.
\item \label{prot:enum:7} If $a=0$, Bob measures the first state in the com\-pu\-ta\-tion\-al basis and checks whether the result is $(i,x_i)$. If this first check fails, he caught Alice cheating. He then checks whether the second state is $(\ket{0}\ket{d}+\ket{i}\ket{x_i})/\sqrt{2}$, by projecting onto this state. If the second check fails, Alice has been caught cheating. 

If $a=1$, he exchanges the first and second state and applies the same procedure. 
\end{enumerate}

The QRAM algorithm is defined as follows~\cite{NieChu00,GLM08}.
\begin{definition}\label{def:qram}
Quantum random-access memory (QRAM) is a collection of unitaries $U_{QRAM}(n)$,where $n \in \mathbb{N}$, such that for all states $\ket{i}$ with $i \in \{1,\dots,n\}$ and all states $\ket{x_1,\ldots, x_{n}}$, where $x_i \in \mX$ (referred to as the data registers or memory cells),
\begin{equation*}
    U_{QRAQM}~\ket{i}\ket{0}\ket{x_1,\ldots, x_{n}} = \ket{i}\ket{x_i}\ket{x_1,\ldots, x_{n}}
\end{equation*}
\end{definition}

The authors of~\cite{GLM10} present an attack for a database with 3 entries where query 1 and query 2 admit two distinguishable answers. They also showed that Bob can prevent this attack of Alice by repeating his queries and checking whether Alice answers consistently. In the following, we present a stronger attack that allows Alice to answer consistently and, therefore, also breaks this approach.

We first establish a lower bound on the trace distance between two quantum states that are central to our analysis of the attack’s success probability. The first state represents a database entry that Bob has not queried throughout the protocol, while the second corresponds to an entry that Bob has queried.
\begin{lemma}\label{lem:dist-adv:ot-1-2}
Let $\ket{\phi}:=\frac{1}{\sqrt{|\mX|}} \sum_{x\in \mX}\ket{x}$. Define the two states $\sigma^0:=\ket{\phi}\bra{\phi}$ and $\sigma^1:=\sum_{x\in \mX}\frac{1}{|\mX|}\ket{x}\bra{x}$. Then the following lower bound on the trace distance holds
\begin{equation*}
    \dis(\sigma^0,\sigma^1)\geq 1-\frac{1}{|\mX|}\;.
\end{equation*}
\end{lemma}
\begin{proof}
We consider the projective measurement defined by the two orthogonal projectors $P
_0:=\vecstate{\phi}$ and $P_1:=\id-\vecstate{\phi}$. When we apply this measurement to $\sigma^0$, the probability $p_0$ of obtaining the result $0$ is $\tr(P_0\sigma^0) = \tr(P_0\vecstate{\phi}) = \tr(\vecstate{\phi}\vecstate{\phi}) = 1$. Measuring $\sigma^1$, the probability $p_1$ to obtain the outcome $1$ is $1-\tr(P_0\sigma^1) = 1-1/|\mX|$. Thus, \eqref{eq:op-meaning-trace-distance} implies that 
$\tfrac{1}{2}\cdot(1+\dis(\sigma^0,\sigma^1))\geq\tfrac{1}{2}(p_0+p_1)= 1-\tfrac{1}{2|\mX|}$.
\end{proof}

Lemma~\ref{lem:dist-adv:ot-1-2} implies the following lower bound on the trace distance of the two states of the database after the execution of the quantum private queries protocol with inputs $i \neq j \in \{1,\dots,n\}$.

\begin{lemma}\label{lem:trace-distance:database}
The trace distance of the two states $\rho^i_{n}$ and $\rho^{j}_{n}$ after the execution of the quantum private queries protocol~\cite{GLM08} with indices $i \neq j \in \{1,\dots,n\}$ is 
\begin{equation*}
    \dis(\rho^i_{n}, \rho^{j}_{n}) \geq 1-\frac{1}{\max\{|\mX_i|,|\mX_j|\}}\;,
\end{equation*}
where $\mX_i$ and $\mX_j$ are the sets of valid answers for the queries $i$ and $j$ respectively.
\end{lemma}
\begin{proof}
Let $\mX_k$ be the set of valid answers for query $k$ for all $k \in \{1,\dots,n\}$. Note that Alice can always add a fixed dummy value $\ket{d}$ as an additional input. Let Alice choose her database entries randomly from the set of valid answers for each query and keep the purification. W.l.o.g we can assume that $|\mX_i| \geq |\mX_j|$. Then the initial state $\ket{\phi^i}_{X_iX'_i}$ of each data register $i$ of Alice is a pure state 
\begin{equation*}
\ket{\phi^i}_{X_iX'_i} = \sum_{x \in \mX_i}\sqrt{P_X(x)}\ket{xx}_{X_iX'_i}=\sum_{x \in \mX_i}\frac{1}{\sqrt{|\mX_i|}}\ket{xx}_{X_iX'_i}\;,
\end{equation*}
where we used the notation that system $X_i$ represents the database entry and system $X'_i$ is its purification. Register $X_0$ stores the dummy entry. The whole database is initially in the state 
\begin{equation*}
    \ket{\phi}_{X_0X_1\ldots X_n} = \ket{d}_{X_0}\kron\bigotimes_{k=1}^n\ket{\phi^i}_{X_kX'_k}\;.
\end{equation*}

Bob prepares the states $\ket{\psi_0}=\ket{i}$ and $\ket{\psi_1}=(\ket{0}+\ket{i})/\sqrt{2}$. Then he chooses a random bit $a$ and sends first the state $\ket{\psi_a}$ and then the state $\ket{\psi_{1-a}}$. Alice computes the answers to both queries using the QRAM algorithm. The algorithm only accesses the data register $X_0$ (dummy entry) and $X_i$ (the answer to Bob's query). 

The state $\ket{\Psi_{AB}}$ before Alice applies the {QRAM} algorithm, describing both Alice's database and Bob's queries, is (assuming w.l.o.g.\ $a=0$)
\begin{align*}
    \ket{\Psi_{AB}}&=(\ket{i}_{q_1}\ket{0}_{a_1}
    \frac{(\ket{0}+\ket{i})_{q_2}}{\sqrt{2}}\ket{0}_{a_2}\ket{d}_{X_0}\kron\bigotimes_{k=1}^n\ket{\phi^i}_{X_kX'_k}\\
  &= \ket{i}_{q_1}\ket{0}_{a_1}\frac{(\ket{0}+\ket{i})_{q_2}\ket{0}_{a_2}}{\sqrt{2}}\ket{d}_{X_0}\kron\bigotimes_{k=1}^n\left(\sum_{x \in \mX_k}\frac{1}{\sqrt{|\mX|}}\ket{xx}_{X_kX'_k}\right)
\end{align*}
The state $\ket{\Psi_{AB}'}$ \emph{after} Alice has applied the {QRAM} algorithm is (we take the sum over the $i$-th register to the front)
\begin{align*}
    \ket{\Psi_{AB}'}&=
       \sum_{x \in \mX_i}  \left((\ket{i}_{q_1}\ket{x_i}_{a_1}\frac{(\ket{0}\ket{d}+\ket{i}\ket{x_i})_{q_2,a_2}}{\sqrt{2}})\ket{d}_{X_0}
\frac{1}{\sqrt{|\mX_i|}}\ket{xx}_{X_iX'_i}\right)\kron\bigotimes_{k\neq i}\ket{\phi^k}_{X_kX'_k}\;.
\end{align*}

Now Bob will perform the measurement $\{\vecstate{i} \kron \vecstate{x_i}\}_{x_i}$ on register $q_1a_1$ and then depending on the outcome $\{\frac{(\ket{0}\ket{d}+\ket{i}\ket{x_i})_{q_2,a_2}}{\sqrt{2}}\frac{(\bra{0}\bra{d}+\bra{i}\bra{x_i})_{q_2,a_2}}{\sqrt{2}}\}$ on $q_2a_2$. 

Bob obtains the outcome $x_i$ with probability $\nicefrac{1}{|\mX_i|}$ for each value. The second check will always succeed and conditioned on outcome $x_i$ and both checks passing, the conditional state $\ket{\Psi_{A}''}$ on Alice's side of the remaining registers is
\begin{align*}
    \ket{\Psi_{A}''}&=
\left(\ket{d}_{X_0}
\ket{xx}_{X_iX'_i}\right)
    \kron\bigotimes_{k\neq i}\ket{\phi^k}_{X_kX'_k}
\end{align*}
Therefore, the mixed state on Alice's side after Bob's measurement (but ignoring the value of the result) is 
\begin{equation*}
    \rho^i_{n} = \ket{d}\bra{d}_{X_0}\kron\bigotimes_{k = 0}^{i-1}\vecstate{\phi^k}_{X_kX'_k}\kron\sum_{x \in \mX_i}\frac{1}{|\mX_i|}\ket{xx}\bra{xx}_{X_iX'_i}
\kron\bigotimes_{k = i+1}^{n-1}\vecstate{\phi^k}_{X_kX'_k}\;.
\end{equation*} 
If we only consider register $X_i$ and ignore the other subsystems, we recover the two states of Lemma~\ref{lem:dist-adv:ot-1-2}, i.e.,
\begin{align*}
    \ptrs{X_0X_1\ldots X_{i-1}X_{i+1}\ldots X_n}(\rho^{i}_{n})=\sigma^{1}_{n}~\text{and}~
    \ptrs{X_0X_1\ldots X_{i-1}X_{i+1}\ldots X_n}(\rho^j_{n})=\sigma^{0}_{n} \;.
\end{align*}
We can use the fact that the partial trace cannot increase the trace distance (first inequality) and apply Lemma~~\ref{lem:dist-adv:ot-1-2} (second inequality) to obtain
\begin{align*}
\dis(\rho^{i}_{n},\rho^{j}_{n})&\geq
    \dis\left(\ptrs{X_0X_1\ldots X_{i-1}X_{i+1}\ldots X_n}(\rho^{i}_{n}),
    \ptrs{X_0X_1\ldots X_{i-1}X_{i+1}\ldots X_n}(\rho^j_{n})\right)=\dis(\sigma^{1}_{n},\sigma^{0}_{n})
    \geq 1- \frac{1}{|\mX_i|}\;.
\end{align*}
\end{proof}

\subsection{Proof of the success rate of the generic attack retrieving all database entries}

In this section of the appendix, we compute the probability that a cheating Bob can obtain \emph{all} database entries in an oblivious transfer (symmetric private information retrieval) protocol. 

In the following, we will show that for any quantum private queries protocol that is $\eps$-correct and is weakly $\eps$-secure for Bob, there exists an attack that allows Bob to compute $m$ database entries for any $2 \leq m \leq n$. The success probability of the attack depends on the number $m$ of database entries that Bob learns. Since Bob can potentially learn all database entries, this attack can violate any sensible definition of security for Alice, even a very weak definition that only prevents a malicious Bob from learning all database entries.

The proof will use the gentle measurement lemma~\cite{Winter99,Wilde13} and Uhlmann's theorem~\cite{Uhlman76}. First, we consider two pure states $\vecstate{\psiAB^{0}}$ and $\vecstate{\psiAB^{1}}$ of Alice and Bob. If the two marginal states, $\rho_A^0$ and $\rho_A^1$, are close in trace distance, then Uhlmann's theorem guarantees the existence of a unitary acting on Alice's side only that transforms the first state $\vecstate{\psiAB^{0}}$ to a state that is close to the second state $\vecstate{\psiAB^{1}}$. The following lemma shows that this also holds when Alice and Bob share two states that are pure conditioned on all the classical information shared between the two parties (classical communication). The proof of the lemma can be found in~\cite{WTHR11}.
\begin{lemma}\label{lem:classical-attack}
For $b \in \{0,1\}$, let 
\[\rho_{XX'AB}^b=\sum_xP_b(x)\vecstate{x}_{X}\kron \vecstate{x}_{X'}\kron \vecstate{\psiAB^{x,b}}\]
with $\dis(\rho_{X'B}^0,\rho_{X'B}^1)\leq \eps$. Then there exists a unitary $U_{AX}$ such that 
\[\dis(\rho_{XX'AB}'^1,\rho_{XX'AB}^1)\leq \sqrt{2 \eps} \]
where $\rho_{XX'AB}'^1=(U_{XA}\kron \id_{X'B})\rho_{XX'AB}^0(U_{XA}\kron \id_{X'B})^\dagger$.
\end{lemma}

If we can predict the outcome of a measurement applied to a quantum state almost with certainty, then this measurement does not disturb the quantum state too much. Therefore, using the state after the measurement for further processing instead of the original state is guaranteed to lead to (almost) the same results. The gentle measurement lemma captures this intuition and gives a quantitative upper bound on the disturbance caused by the measurements.
\begin{lemma}[Gentle measurement lemma]\cite{Winter99, Wilde13}\label{lem:gentle-measurement}
Consider a density operator $\rho$ and an element of a measurement $\mE_x$. Suppose that the measurement has a high probability of detecting state $\rho$, i.e., we have $\tr(\mE_x\rho) \geq 1-\eps$. Then the post-measurement state 
\begin{equation*}
   \rho'=\frac{\sqrt{\mE_x}\rho\sqrt{\mE_x}^{\dagger}}{\tr(\mE_x\rho)} 
\end{equation*} 
is $\sqrt{\eps}$-close to the original state $\rho$, i.e., 
$\dis(\rho,\rho') \leq \sqrt{\eps}$.
\end{lemma}
Lemma~\ref{lem:gentle-measurement} implies the following upper bound for the disturbance caused by a measurement to a classical-quantum state.
\begin{lemma}\label{lem:gentle-measurement-cq}
Consider a cq-state $\rho_{XA}=\sum_{x \in \mX} P_X(x)\ket{x}\bra{x}\kron \rho_{A}^x$ and measurement acting on system $A$ defined by a trace-preserving completely positive map $\mathcal{E}$ with $\mathcal{E}:\sigma_{XA} \mapsto \sum_x \ket{x}\bra{x}_{X'}\kron\mathcal{E}_x\sigma_{XA}\mathcal{E}_x^\dagger$. Suppose that for each $x \in \mX$ the measurement operator $\mathcal{E}_x$ has a high probability of detecting state $\rho_{A}^x$, i.e., we have $\tr(\mE_x\mathcal\rho_A^x) \geq 1-\eps$ where $\mE_x=M_x^\dagger M_x$. Then for the state $\rho_{X'XA}'=\mE(\rho_{XA})$ we have that $\Pr[X' = X]\geq 1-\eps$ and the post-measurement state, ignoring the measurement outcome, 
\begin{equation*}
\rho_{XA}'=\sum_x(\id_X \kron M_x)\rho_{XA}(\id_X\kron M_x)^\dagger
\end{equation*}is $(\sqrt{\eps}+\eps)$-close to the original state $\rho$, i.e., 
\begin{equation*}
\dis(\rho_{XA},\rho_{XA}') \leq \sqrt{\eps}+\eps\;.
\end{equation*}
\end{lemma}
\begin{proof}
For the probability that the measurement outcome is equal to $X$ it holds that
\begin{align*}
\Pr[X'=X]&=\sum_xP_X(x)\tr(\mE_x\mathcal\rho_A^x)
\geq \sum_xP_X(x)(1-\eps)
=1-\eps\;.
\end{align*}
For any $\hat{x} \in \mX$ we can derive the following upper bound for the disturbance caused by the measurement to the conditional state
\begin{align}\label{ineq:disturbance-single}
\begin{split}
\dis(\rho_A^{\hat{x}},\mE(\rho_A^{\hat{x}}))&= \dis(\rho_A^{\hat{x}},\sum_{x \in \mX}M_x\rho_{A}^{\hat{x}}M_x^\dagger)\\ 
    &\leq\sum_{x \in \mX}\tr(\mE_x\rho_A^{{\hat{x}}})\dis(\rho_A^{\hat{x}},M_x\rho_{A}^{\hat{x}}M_x^\dagger/\tr(\mE_x\mathcal\rho_A^{\hat{x}}))\\
    &=\sum_{x \neq \hat{x} \in \mX}\tr(\mE_x\rho_A^{{\hat{x}}})\dis(\rho_A^{\hat{x}},M_x\rho_{A}^{\hat{x}}M_x^\dagger/\tr(\mE_x\mathcal\rho_A^{\hat{x}}))+\tr(\mE_{\hat{x}}\rho_A^{\hat{x}})\dis(\rho_A^{\hat{x}},M_{\hat{x}}\rho_{A}^{\hat{x}}M_{\hat{x}}^\dagger/\tr(\mE_{\hat{x}}\mathcal\rho_A^{\hat{x}}))\\
    &\leq(1-\tr(\mE_{\hat{x}}\rho_A^{\hat{x}}))+\tr(\mE_{\hat{x}}\rho_A^{\hat{x}})\sqrt{\eps}\\
    &\leq \eps+\sqrt{\eps}\;,
\end{split}
\end{align}
where we used the convexity of the trace distance in the first inequality and Lemma~\ref{lem:gentle-measurement} in the second inequality. This implies that the total disturbance caused by the measurement is at most
\begin{align*}
    \dis(\rho_{XA},\rho_{XA}') &= \dis(\rho_{XA}, \mathcal{E}(\rho_{XA)})\\
    &= \dis(\sum_{x \in \mX}P_{X}(x)\ket{x}\bra{x}\kron\rho_{A}^x,\sum_{x \in \mX}P_{X}(x)\ket{x}\bra{x}\kron\mathcal{E}(\rho_{A}^x))\\
    &= \sum_{x \in \mX}P_{X}(x)\dis(\rho_{A}^x,\mathcal{E}(\rho_{A}^x))\\
    &\leq\sqrt{\eps}+\eps\;,
\end{align*}
where we used \eqref{eq:trace-distance-cq-states} to obtain the third equality and~\eqref{ineq:disturbance-single} to obtain the inequality.
\end{proof}

We can now formulate the main technical statement.
\mainresult*
\begin{proof}
We assume that Alice, the owner of the database, chooses all her database entries $X_i$ uniformly at random from the set of valid answers $\mX_i$. For all $i \in \{1,\dots,n\}$, we can assume that the state $\rho_{AB}^i$ at the end of the protocol is pure given all the classical communication.

The correctness of the protocol implies that Bob, the user of the database, can compute $X_i$ from $\rho^i_{AB}$ with probability $1-\eps$. From Lemma~\ref{lem:gentle-measurement-cq} we know that the resulting state $\Tilde{\rho}^i_{AB}$ after computing $X_i$ is $(\sqrt{\eps}+\eps)$-close to $\rho^i_{AB}$, i.e., 
\begin{equation}\label{eq:damage-measurement}
\dis(\rho_{AB}^i, \Tilde{\rho}_{AB}^i) \leq \sqrt{\eps}+\eps\;.
\end{equation}
Since the protocol is secure for an honest Bob, we have that for all $i,j \in \{1,\dots,n\}$
\begin{equation*}
\dis(\rho_{A}^i, \rho_A^j) \leq 2\eps\;.
\end{equation*}
With Lemma \ref{lem:classical-attack} this implies that there exists a unitary $U_{B,i,j}$ such that
\begin{equation}\label{ineq:imposs:ot-full:uhlmann}
\dis\big(\rho_{AB}'^j, \rho_{AB}^j\big)\leq 2\sqrt{\eps}\;,
\end{equation}
where $\rho_{AB}'^j= (\id \kron U_{B,i,j})\rho_{AB}^i(\id \kron U_{B,i,j})^\dagger$. Let $\mathcal{E}_i$ be the operation that computes the value $X_i$ from Bob's system $B$ and let $\Tilde{\rho}_{AB}^i = \mathcal{E}_i(\rho_{AB}^i)$ be the resulting state. 
Then, omitting the identity on $A$, it holds that
\begin{align*}
\begin{split}  \dis(U_{B,i,j}\mathcal{E}_i(\rho_{AB}^i)U_{B,i,j}^\dagger, \rho_{AB}^j)
       &=  \dis(U_{B,i,j}\Tilde{\rho}_{AB}^iU_{B,i,j}^\dagger, \rho_{AB}^j)\\
      &\leq  \dis(U_{B,i,j}\Tilde{\rho}_{AB}^iU_{B,i,j}^\dagger, U_{B,i,j}\rho_{AB}^iU_{B,i,j}^\dagger)+ \dis(U_{B,i,j}\rho_{AB}^iU_{B,i,j}^\dagger, \rho_{AB}^j)\\
      &= \dis( \Tilde{\rho}_{AB}^i, \rho_{AB}^i) + \dis\big(\rho_{AB}'^j, \rho_{AB}^j\big)\\
      &\leq (\sqrt{\eps}+\eps)+2\sqrt{\eps} \\
      &=  3\sqrt{\eps}+\eps\;,
\end{split}
\end{align*}
where the first inequality follows from the triangle inequality, the next equality from  the fact that the trace distance is preserved under unitary transformations and the second inequality follows from~\eqref{eq:damage-measurement} and~\eqref{ineq:imposs:ot-full:uhlmann}.

We have shown that for all $i,j \in \{1,\dots,n\}$ there exists quantum operation $S_{i,j}$, namely first measuring the value $X_i$ and then applying the unitary $U_{B,i,j}$, that computes the value $X_i$ from the state $\rho_{AB}^i$ and results in a state that is $(3\sqrt{\eps}+\eps)$-close to $\rho_{AB}^j$, i.e,
\begin{equation}\label{eq:damage-operation}
    \dis(S_{i,j}(\rho_{AB}^i), \rho_{AB}^j) \leq 3\sqrt{\eps}+\eps\;.
\end{equation}
Thus, the attack allows Bob to compute all of Alice's inputs: Bob has input $1$ and honestly follows the protocol except from keeping his state purified. At the end of the protocol, he computes successively all values $X_j$ using the above operation $S_{i,i+1}$, which means that  he computes the next value $X_{i}$ and then applies a unitary to rotate the state close to the next target state $\rho_{AB}^{i+1}$. Let $S_i:=S_{i-1,i}$. Next, we will show that 
\begin{equation}\label{eq:damage-succession}
    \dis(S_l(S_{l-1}(\ldots S_2(\rho_{AB}^1))), \rho_{AB}^l)\leq (l-1)(3\sqrt{\eps}+\eps)\;.
\end{equation}
From~\eqref{eq:damage-operation}, we know that the statement holds for $l=2$ 
\begin{equation*}
    \dis(S_2(\rho_{AB}^1),\rho_{AB}^2)\leq 3\sqrt{\eps}+\eps\;.
\end{equation*}
Assume that the statement holds for $l-1$, then we have for any $2 \leq l \leq n$
\begin{align*}
\begin{split}
    \dis(S_l(S_{l-1}(\ldots S_2(\rho_{AB}^1))), \rho_{AB}^l)&\leq\dis((S_l(S_{l-1}(\ldots S_2(\rho_{AB}^1))), S_l(\rho_{AB}^{l-1}))+\dis( S_l(\rho_{AB}^{l-1}),\rho_{AB}^l)\\
    &\leq\dis((S_{l-1}(\ldots S_2(\rho_{AB}^1)), \rho_{AB}^{l-1})+\dis( S_l(\rho_{AB}^{l-1}),\rho_{AB}^l)\\
    &\leq (l-2)(3\sqrt{\eps}+\eps) + 3\sqrt{\eps}+ \eps\\
    &= (l-1))(3\sqrt{\eps}+\eps)\;,
    \end{split}
\end{align*}
where we first applied the triangle inequality, then we used the fact that quantum operations cannot increase the trace distance in the second inequality, and finally we used~\eqref{eq:damage-operation} and the assumption that~\eqref{eq:damage-succession} holds for $l-1$ in the third inequality.

Let $\eps_l$  be the probability that Bob fails to compute the $l$-th value correctly. We know that Bob can compute the value $X_j$ from the state $\rho_{AB}^j$ with probability at least $1-\eps$ and that $\eps_1 \leq \eps$. We can upper bound the probability for the other $\eps_l$ adding the distance from the state $\rho_{AB}^l$ given by~\eqref{eq:damage-succession} and the error $\eps$, i.e.,
\begin{align}\label{ineq:prob-fail-step}
    \eps_l &\leq (l-1)(3\sqrt{\eps}+\eps)+\eps
    \leq l\cdot(3\sqrt{\eps}+\eps)\;
\end{align}
Thus, we can apply the union bound to obtain the claimed upper bound for the probability that Bob fails to compute any of the $m \geq 2$ values correctly
\begin{align*}
\sum_{l=1}^{m}\eps_l&\leq \eps_1 +\sum_{l=2}^{m}l\cdot(3\sqrt{\eps}+\eps)\\
&\leq\eps+(3\sqrt{\eps}+\eps)\sum_{l=2}^{m}l\\
&=\eps+(3\sqrt{\eps}+\eps)(m(m+1)/2-1)\\
&=\frac{1}{2}(3\sqrt{\eps}+\eps)(m^2+m)-3\sqrt{\eps}\;.
\end{align*}
We can in the following assume that $\eps \leq \frac{1}{64}$. Otherwise $1-2m^2\sqrt{\eps}\leq 0$ for $m\geq 2$ and the statement of the theorem is trivially true. Let $m \geq 4$. Since $\eps \leq \sqrt{\eps}/8$ and $m \leq m^2/4$, it holds that 
\begin{align*}
\frac{1}{2}(3\sqrt{\eps}+\eps)(m^2+m)-3\sqrt{\eps} &\leq \frac{1}{2}(3\sqrt{\eps}+\sqrt{\eps}/8)(m^2+m^2/4)\\
&=\tfrac{125}{64}\cdot m^2\sqrt{\eps}\\
&\leq 2m^2\sqrt{\eps}\;.
\end{align*}
It can easily be checked that $\frac{1}{2}(3\sqrt{\eps}+\eps)(m^2+m)-3\sqrt{\eps}\leq 2m^2\sqrt{\eps}$ for $m\in \{2,3\}$. This implies the claimed statement for all $2 \leq m \leq n$.
\end{proof}
In the following, we will prove that there exists a constant $c >0$ such that there exists no quantum private queries protocol for databases with $n \geq 2$ entries that is $\eps$-correct according to Correctness (Definition~\ref{def:correctness}) and  $\eps$-secure according to User Privacy (Definition~\ref{def:weak-user-privacy}) and Data Privacy (Definition~\ref{def:weak-database-privacy}) if $\eps\leq c$. 
\begin{corollary}\label{cor:imposs-qpq}
Consider a database with $n \geq 2$ entries, and let $\mX_i$ be the set of valid answers for query $i$ for all $1 \leq i \leq n$. Let there be at least two queries with multiple valid answers, i.e., there exist $i \neq j \in \{1,\ldots,n\}$ with $\min\{|\mX_i|,|\mX_j|\}\geq 2$. If $\eps \leq \frac{1}{64}$, then there is no quantum private queries protocol for this database that is $\eps$-correct and $(\eps,\delta)$-secure for the user and for the owner of the database for any  $\eps \leq \delta \leq 1$.
\end{corollary}
\begin{proof}
We consider the attack of Theorem~\ref{thm:main-result}. Let $i \neq j \in \{1,\ldots,n\}$ with $\min\{|\mX_i|,|\mX_j|\}\geq 2$. W.l.o.g we can assume that $i=1$ and $j=2$. Let $\eps_1$ and $\eps_2$ as defined in the proof of Theorem~\ref{thm:main-result} be the probabilities that Bob fails to compute the first and second entry of the database correctly. The correctness of the protocol implies that $\eps_1 \leq \eps$. From equation~\eqref{ineq:prob-fail-step}, we know that $\eps_2 \leq 3\sqrt{\eps}+2\eps$. Thus, the probability that Bob computes both entries correctly must be at least $1-\eps_1-\eps_2\geq 1-3(\sqrt{\eps}+\eps)$. Since the protocol is $\eps$-secure for the owner of the database, it must hold that $1-3(\sqrt{\eps}+\eps) \leq \frac{1}{2}+\eps$, which implies the claimed statement.
\end{proof}
\remark \label{rem:validity} As explained above, we consider an answer to a query valid if Bob accepts it. Without additional assumptions, Bob accepts any answer supported by the protocol. For example, in the protocol analyzed in~\cite{GLM10}, there are $n$ queries, each with an $m$-bit answer. In this case, the set of valid answers for query $i \in \{1,\ldots,n\}$ is simply the set of all $m$-bit strings, i.e., $\mX_i = \{\ket{0}, \ldots, \ket{2^m-1}\}$.

The security analysis in~\cite{GLM10} further assumes that each query has a \emph{unique} valid answer. This condition is satisfied if every query admits a unique correct answer and Bob can independently verify the correctness of any received answer. Under this assumption, the set of valid answers for each query reduces to a single $m$-bit string. 

More generally, one may assume that Bob can verify the correctness of answers, but that each query may admit a set of correct answers of arbitrary cardinality—that is, any subset of the answers supported by the protocol. Note that Theorem~\ref{thm:main-result} and Corollary~\ref{cor:imposs-qpq} apply under both notions of validity. In the proof of Theorem~\ref{thm:main-result}, if Bob can verify correctness, Alice prepares a database containing a superposition of all \emph{correct} answers for each query. Without this assumption, Alice instead prepares a database containing a superposition of all answers supported by the protocol.

\remark \label{rem:imposs-qpq} Corollary~\ref{cor:imposs-qpq} rules out the existence of secure quantum private query protocols for all databases with at least two queries with multiple valid answers. This leaves open only the case where at most one query admits multiple valid answers. In this case, a computationally unbounded user can reconstruct the entire database—except for the entry corresponding to that single ambiguous query—by computing the unique valid answers to all other queries. Consequently, a dishonest user can recover the full database with just a single query, eliminating any non-trivial security guarantee for the database owner.

If, however, we assume the user of the database is computationally bounded, then the problem can be solved using computationally secure protocols. In such cases, there is no need to resort to cheat-sensitive protocols, as discussed earlier. It is also important to note that the impossibility result remains valid even for protocols that restrict the \emph{honest} database owner to classical database entries. The attack only requires that a \emph{dishonest} database owner be able to randomly select entries from the set of answers accepted by the user.

\remark \label{rem:imposs-qpq-countermeasures} As discussed above, the main security analysis in~\cite{GLM10} relies on the assumption that each query has a \emph{unique} valid answer. If this assumption is dropped and multiple answers are valid for the same query, the security analysis in~\cite{GLM10} does not apply. As a remedy, it is proposed in~\cite{GLM10} to prevent the owner of the database from cheating by requiring her to send a composite answers that contains all possible answers to a query in a pre-established order. In this way, each query has a unique valid answer again. If, however, each query has a unique valid answer, there can be no security guarantee for the owner of the database against a computationally unbounded user, as explained in Remark~\ref{rem:imposs-qpq}. If the user is computationally bounded, there is no need to resort to cheat-sensitive security because there are protocols~\cite{BCKM21} that solve the problem in a non-cheat-sensitive setting, as discussed in more detail earlier. 

As an alternative countermeasure to prevent attacks on databases with multiple valid answers for the same query, the authors of \cite{GLM10} require the user of the database to repeatedly send the same query and check that the owner answers consistently. This countermeasure does not prevent our generic attack because the owner of the database can indefinitely give consistent answers to repeated queries. Therefore, Theorem~\ref{thm:main-result} also applies to protocols that require the user to repeatedly send the same query, and Corollary~\ref{cor:imposs-qpq} rules out that any such protocol is secure. 

\end{document}